\documentclass{article}
\usepackage[utf8]{inputenc}
\usepackage[letterpaper, margin=1in]{geometry}

\usepackage{setspace}
\usepackage{multirow}
\usepackage{graphicx}
\usepackage{dsfont}
\usepackage{amsfonts}
\usepackage{amsmath}
\usepackage{bbm,bbold}
\usepackage{float}
\usepackage[normalem]{ulem}
\usepackage[dvipsnames]{xcolor}
\definecolor{DartmouthGreen}{RGB}{0, 105, 62}
\usepackage{physics}
\usepackage[utf8]{inputenc}

\usepackage{graphicx}
\usepackage{subcaption}
\usepackage{mwe}

\usepackage{amsthm}
\usepackage{authblk}

\newtheorem{theorem}{Theorem}[section]

\usepackage{mathtools}
\usepackage{amssymb}
\usepackage{tikz}
\usetikzlibrary{quantikz}
\usetikzlibrary{shapes,snakes}

\usepackage{pgfplots}
\pgfplotsset{compat=1.13}

 \usepackage{graphicx}
 \usepackage[colorlinks=true, linkcolor=green, citecolor=blue]{hyperref}

 \def\be{\begin{equation}}
 \def\ee{\end{equation}}
 \def\bea{\begin{eqnarray}}
 \def\eea{\end{eqnarray}}
 \def\bean{\begin{eqnarray*}}
 \def\eean{\end{eqnarray*}}
 \def\gsim{\mathrel{\rlap{\lower0.2em\hbox{$\sim$}}\raise0.2em\hbox{$>$}}}
 \def\ksim{\mathrel{\rlap{\lower0.2em\hbox{$\sim$}}\raise0.2em\hbox{$<$}}}
 \def\kg{\mathrel{\rlap{\lower0.25em\hbox{$>$}}\raise0.25em\hbox{$<$}}}

 \newcommand\glnf[1]{$\mathrm{GL}_n\left(\mathbb{F}_2\right)$}
 
 \newtheorem{lemma}[theorem]{Lemma}
 
\begin{document}
\title{\vspace{-5.5mm} A Sierpinski Triangle Data Structure for Efficient Array Value Update and Prefix Sum Calculation}

%--------------------------------------------------------------------------
\author{Brent Harrison, Jason Necaise, Andrew Projansky, James D.\ Whitfield}
\affil{Department of Physics and Astronomy, Dartmouth College, Hanover, New Hampshire 03755, USA}

\renewcommand\Affilfont{\itshape\small}
\date{}

\maketitle
\begin{abstract}
    The binary indexed tree, or Fenwick tree, is a data structure that can efficiently update values and calculate prefix sums in an array. It allows both of these operations to be performed in $O(\log_2 N)$ time. Here we present a novel data structure resembling the Sierpinski triangle, which accomplishes these operations with the same memory usage in $O(\log_3 N)$ time instead. We show this order to be optimal by making use of a connection to quantum computing.
\end{abstract}

\section{Introduction: Fenwick Trees}

The Fenwick tree~\cite{fenwick_new_1994, fenwick_new_1996} is a binary tree data structure often used for storing frequencies and manipulating cumulative frequency tables. Interestingly, it has also seen applications in the quantum simulation of fermionic systems~\cite{havlicek_operator_2017}. In general, it is useful for applications that require storing and dynamically updating data in an array, as well as efficiently calculating prefix sums. The ``prefix sum'' or ``inclusive scan'' operation on a sequence of bits $n_0, n_1, n_2, \dots$ outputs a second sequence $p_0, p_1, p_2, \dots$ where
\begin{equation}
p_j \equiv \bigoplus_{i<j} n_i
\end{equation}
is the binary sum of the first $j+1$ bits. A na{\"i}ve array implementation allows data updates in $O(1)$ time, but requires $O(N)$ time for prefix sum calculation. The Fenwick tree has the advantage of performing both operations in $O(\log_2 N)$ time. This tree is constructed by the following algorithm,
$\text{}$\newline \\
define {\bf{fenwick}}$(S, E)$:\\
$\text{} \hspace{7mm} \text{if}$ $S \neq E$:\\
$\text{} \hspace{14mm}$        connect $E$ to $\lfloor\frac{E+S}{2}\rfloor$;\\
$\text{} \hspace{14mm}$        {\bf fenwick}$\left(S, \lfloor\frac{E+S}{2}\rfloor\right)$;\\
$\text{} \hspace{14mm}$        {\bf fenwick}$\left(\lfloor\frac{E+S}{2}\rfloor,E\right)$;\\
$\text{} \hspace{7mm}$    else:\\
$\text{} \hspace{14mm}$        return;
\newline \\
where the function {\bf fenwick}$(0,\ N-1)$ generates a Fenwick tree with $N$ nodes. See Fig.~\ref{fig:fenwickTreeN7} for a seven-node example. The tree represents the data structure as follows.

\begin{figure}
    \centering
    \includegraphics[width=0.9\linewidth]{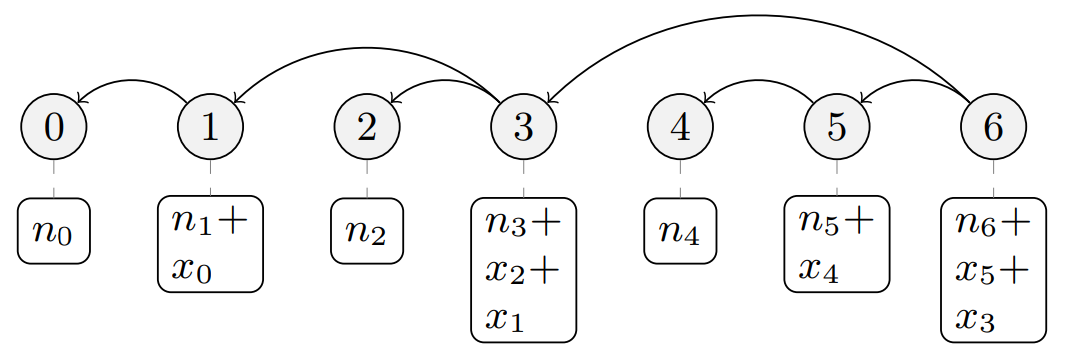}
    \caption{A data structure based on a seven-node Fenwick tree. Source:~\cite{havlicek_operator_2017}}
    \label{fig:fenwickTreeN7}
\end{figure}

We begin with a binary array of $N$ elements $n_j$. We encode this information in a new array of $N$ elements $x_j$, defined recursively as

\begin{equation}\label{eq:nx}
x_j = n_j + \sum_{k \in F(j)}x_k,
\end{equation}

where $F(j)$ is the set of children of node $j$ in the tree. It is easily seen that updating an element $n_j$ of the original array requires multiple bitflips in the encoded array - we must update $x_j$ as well as all its ancestors on the tree. However, calculating prefix sums is simplified. In our seven-node example, in order to calculate $p_6$, we need only query bit 6 and its children on the tree, i.e. $p_6 = x_6 + x_5 + x_3$. As the Fenwick tree is well-known~\cite{fenwick_new_1994, fenwick_new_1996}, we will not prove that these operations are $O(\log_2 N)$ here.

We will now describe a novel data structure similar to the Fenwick tree, but with a structure resembling the Sierpinski triangle, which accomplishes the array update and prefix sum operations in $O(\log_3 N)$ time. We will call this structure the Sierpinski tree.

\section{Construction of the Sierpinski Tree}\label{sec:construction}

We define our data structure in terms of a directed tree with $N$ nodes. We will present an algorithm for the construction of this tree when $N$ is a power of 3, which we will call a ``full'' tree. In order to construct the tree for other values of $N$, follow the procedure to construct a ``full'' tree\footnote{In general, the result will be a forest, rather than a single tree. In an abuse of nomenclature, we will nonetheless refer to it as a tree throughout.} with $3^{\lceil \log_3(N) \rceil}$ nodes, and then delete the nodes with indices $i \geq N$. The algorithm follows:
$\text{}$\newline \\
define {\bf{sierpinski}}$(S, E)$:\\
$\text{} \hspace{7mm} \text{if}$ $S \neq E$: {\color{teal} // start $\neq$ end }\\
$\text{} \hspace{14mm}$ L = $S + \frac{1}{2}(\frac{E-S+1}{3}-1)$; {\color{teal} // ``left'' point, midpoint of first third of $[S,E]$)}\\
$\text{} \hspace{14mm}$ C = $\frac{S+E}{2}$; {\color{teal} // ``center'' point, midpoint of second third of $[S,E]$}\\
$\text{} \hspace{14mm}$ R = $E - (L-S)$; {\color{teal} // ``right'' point, midpoint of final third of $[S,E]$}\\
$\text{} \hspace{14mm}$        connect $C$ to $L$;\\
$\text{} \hspace{14mm}$        connect $C$ to $R$;\\
$\text{} \hspace{14mm}$        {\color{teal}// Divide interval into thirds, apply function to each third}\\
$\text{} \hspace{14mm}$ T = $2(L-S)$;\\
$\text{} \hspace{14mm}$        {\bf sierpinski}$(S,\ S + T)$;\\
$\text{} \hspace{14mm}$        {\bf sierpinski}$(S + T + 1,\ S + 2T + 1)$;\\
$\text{} \hspace{14mm}$        {\bf sierpinski}$(S + 2T + 2,\ E)$;\\
$\text{} \hspace{7mm}$    else:\\
$\text{} \hspace{14mm}$        return;

Here the function {\bf sierpinski}$(0,3^k - 1)$ will create a Sierpinski tree with $N = 3^k$ nodes. Figs.~\ref{fig:n9} and \ref{fig:n27} show the tree for $N = 9$ and $N = 27$ respectively. 

% \begin{figure}[h]
%     \centering
%     \includegraphics[width = 0.3\linewidth]{n9.png}
%     \caption{The Sierpinski tree for $N = 9$}
%     \label{fig:n9}
% \end{figure}

\begin{figure}[h]
  \centering
  \scalebox{0.67}{
    \begin{tikzpicture}
      \draw
        (0, 0) node[circle,draw=black,fill=black!4,minimum size=1.05cm, label=center:\Large 0](0){}
        (1, 2) node[circle,draw=black,fill=black!4,minimum size=1.05cm, label=center:\Large 1](1){}
        (2, 0) node[circle,draw=black,fill=black!4,minimum size=1.05cm, label=center:\Large 2](2){}
        (2, 3) node[circle,draw=black,fill=black!4,minimum size=1.05cm, label=center:\Large 3](3){}
        (3, 5) node[circle,draw=black,fill=black!4,minimum size=1.05cm, label=center:\Large 4](4){}
        (4, 3) node[circle,draw=black,fill=black!4,minimum size=1.05cm, label=center:\Large 5](5){}
        (4, 0) node[circle,draw=black,fill=black!4,minimum size=1.05cm, label=center:\Large 6](6){}
        (5, 2) node[circle,draw=black,fill=black!4,minimum size=1.05cm, label=center:\Large 7](7){}
        (6, 0) node[circle,draw=black,fill=black!4,minimum size=1.05cm, label=center:\Large 8](8){};
      \begin{scope}[->]
        \draw (1) to (0);
        \draw (1) to (2);
        \draw (4) to[bend right](1);
        \draw (4) to[bend left](7);
        \draw (4) to (3);
        \draw (4) to (5);
        \draw (7) to (6);
        \draw (7) to (8);
      \end{scope}
    \end{tikzpicture}}
    \caption{The Sierpinski tree for $N = 9$}
    \label{fig:n9}
\end{figure}
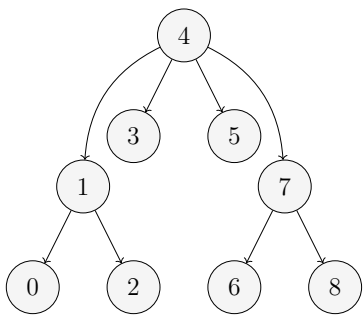

% \begin{figure}[h]
%     \centering
%     \includegraphics[width = 0.5\linewidth]{n27.png}
%     \caption{The Sierpinski tree for $N = 27$}
%     \label{fig:n27}
% \end{figure}

\begin{figure}[h]
  \centering
  \begin{tikzpicture}[scale=0.65]
      \draw
        (0, 0) node[circle,draw=black,fill=black!4,minimum size=0.7cm,label=center:0] (0){}
        (1, 2) node[circle,draw=black,fill=black!4,minimum size=0.7cm,label=center:1] (1){}
        (2, 0) node[circle,draw=black,fill=black!4,minimum size=0.7cm,label=center:2] (2){}
        (2, 3) node[circle,draw=black,fill=black!4,minimum size=0.7cm,label=center:3] (3){}
        (3, 5) node[circle,draw=black,fill=black!4,minimum size=0.7cm,label=center:4] (4){}
        (4, 3) node[circle,draw=black,fill=black!4,minimum size=0.7cm,label=center:5] (5){}
        (4, 0) node[circle,draw=black,fill=black!4,minimum size=0.7cm,label=center:6] (6){}
        (5, 2) node[circle,draw=black,fill=black!4,minimum size=0.7cm,label=center:7] (7){}
        (6, 0) node[circle,draw=black,fill=black!4,minimum size=0.7cm,label=center:8] (8){}
        (4, 7) node[circle,draw=black,fill=black!4,minimum size=0.7cm,label=center:9] (9){}
        (5, 9) node[circle,draw=black,fill=black!4,minimum size=0.7cm,label=center:10] (10){}
        (6, 7) node[circle,draw=black,fill=black!4,minimum size=0.7cm,label=center:11] (11){}
        (6, 10) node[circle,draw=black,fill=black!4,minimum size=0.7cm,label=center:12] (12){}
        (7, 12) node[circle,draw=black,fill=black!4,minimum size=0.7cm,label=center:13] (13){}
        (8, 10) node[circle,draw=black,fill=black!4,minimum size=0.7cm,label=center:14] (14){}
        (8, 7) node[circle,draw=black,fill=black!4,minimum size=0.7cm,label=center:15] (15){}
        (9, 9) node[circle,draw=black,fill=black!4,minimum size=0.7cm,label=center:16] (16){}
        (10, 7) node[circle,draw=black,fill=black!4,minimum size=0.7cm,label=center:17] (17){}
        (8, 0) node[circle,draw=black,fill=black!4,minimum size=0.7cm,label=center:18] (18){}
        (9, 2) node[circle,draw=black,fill=black!4,minimum size=0.7cm,label=center:19] (19){}
        (10, 0) node[circle,draw=black,fill=black!4,minimum size=0.7cm,label=center:20] (20){}
        (10, 3) node[circle,draw=black,fill=black!4,minimum size=0.7cm,label=center:21] (21){}
        (11, 5) node[circle,draw=black,fill=black!4,minimum size=0.7cm,label=center:22] (22){}
        (12, 3) node[circle,draw=black,fill=black!4,minimum size=0.7cm,label=center:23] (23){}
        (12, 0) node[circle,draw=black,fill=black!4,minimum size=0.7cm,label=center:24] (24){}
        (13, 2) node[circle,draw=black,fill=black!4,minimum size=0.7cm,label=center:25] (25){}
        (14, 0) node[circle,draw=black,fill=black!4,minimum size=0.7cm,label=center:26] (26){};
      \begin{scope}[->]
        \draw (1) to (0);
        \draw (1) to (2);
        \draw (4) to[bend right] (1);
        \draw (4) to[bend left] (7);
        \draw (4) to (3);
        \draw (4) to (5);
        \draw (7) to (6);
        \draw (7) to (8);
        \draw (10) to (9);
        \draw (10) to (11);
        % \draw (13) + (-0.54, 0) to[bend right] (4);
        % \draw (13) + (0.54, 0)to[bend left] (22);
        \draw (13) to[out=197.5, in=90] (4);
        \draw (13) to[out=342.5, in=90] (22);
        \draw (13) to[bend right] (10);
        \draw (13) to[bend left] (16);
        \draw (13) to (12);
        \draw (13) to (14);
        \draw (16) to (15);
        \draw (16) to (17);
        \draw (19) to (18);
        \draw (19) to (20);
        \draw (22) to[bend right] (19);
        \draw (22) to[bend left] (25);
        \draw (22) to (21);
        \draw (22) to (23);
        \draw (25) to (24);
        \draw (25) to (26);
      \end{scope}
    \end{tikzpicture}
    \caption{The Sierpinski tree for $N = 27$}
    \label{fig:n27}
\end{figure}
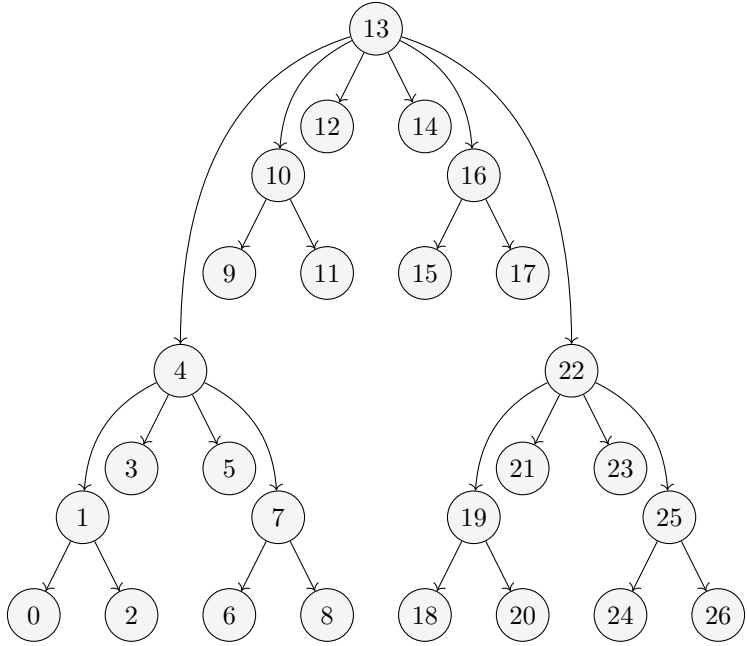

The tree represents the data stored analogously to the Fenwick tree, with a binary array of $N$ elements encoded in a new array of $N$ elements $x_j$, defined recursively as in Eq.~\ref{eq:nx},

\begin{equation}
x_j = n_j + \sum_{k \in F(j)}x_k,
\end{equation}

where $F(j)$ is the set of children of node $j$ on the tree. We illustrate this for $N = 9$ in Fig.~\ref{fig:sierpinski_nx}.

\section{Array Update and Prefix Sum Complexity}

In this section, we will show that the time complexity of the array update and prefix sum operations is $O(\log_3 N)$. We begin with some definitions.

Let an $N$-node Sierpinski tree with least index $j_0$ be defined as $T(j_0, N) \equiv \{V(j_0, N),\ E(j_0, N)\}$, where the set of vertices $V(j_0, N) = \{j_0, j_0 + 1,\dots,j_0 + N-1\}$, and the set of directed edges $E(j_0, N)$ is the corresponding set of ordered pairs defined by the algorithm in Sec.~\ref{sec:construction}. For brevity, we also define $T(N) = \{ V(N), E(N)\} \equiv T(0, N)$. 

It will be useful to discuss the left, central and right subtrees. For a full tree $T(3^k)$, we respectively define them as
\begin{equation}
\begin{aligned}
T_L(3^k) &\equiv \{V_L(3^k), E_L(3^k)\} = T\left(0,\ 3^{k-1}\right),\\
T_C(3^k) &\equiv \{V_C(3^k), E_C(3^k)\} = T\left(3^{k-1},\ 3^{k-1}\right),\\
T_R(3^k) &\equiv \{V_R(3^k), E_R(3^k)\} = T\left(2\cdot 3^{k-1}, 3^{k-1}\right).
\end{aligned}
\end{equation}
whose respective roots have indices
\begin{equation}
\begin{aligned}
L &= \frac{1}{2}(3^{k-1} - 1),\\
C &= \frac{1}{2}(3^{k} - 1),\\
R &= 2C - L.
\end{aligned}
\end{equation}
For an unfull tree with $N$ nodes $T(N)$, we can define the subtrees by constructing the smallest full tree with $3^k > N$ nodes and then deleting nodes with indices $i \geq N$ from the subtrees. 

\begin{figure}
    \centering    \includegraphics[width=0.9\linewidth]{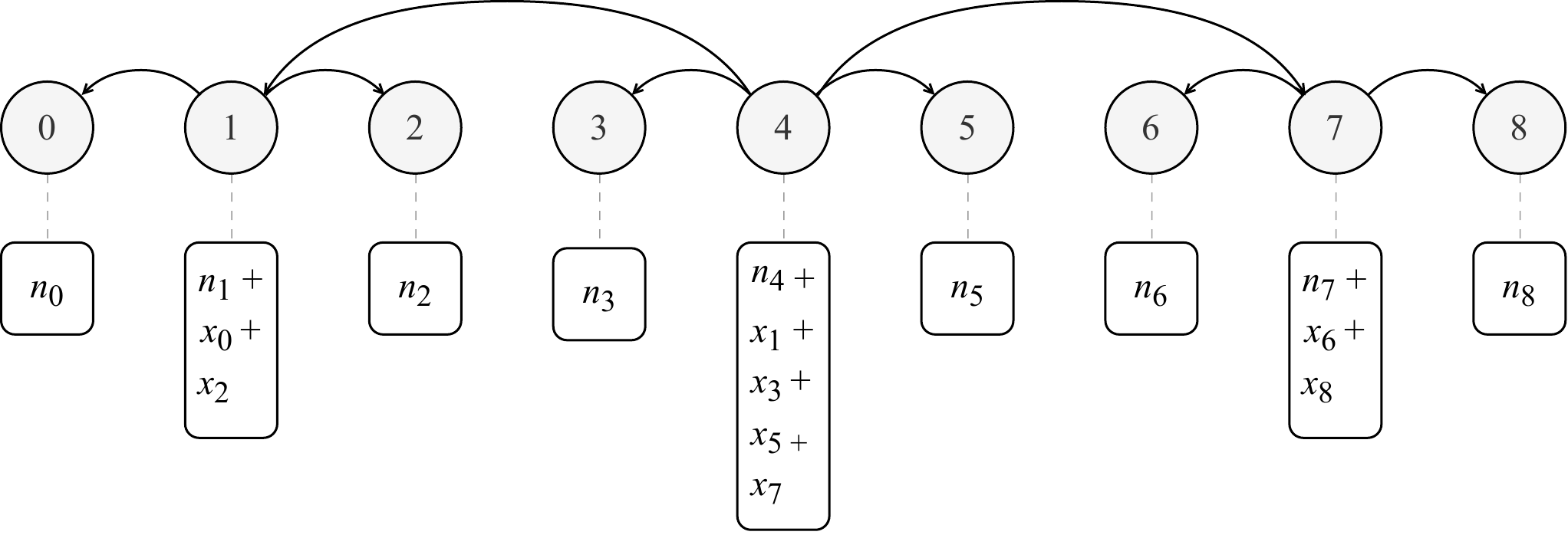}
    \caption{Diagram showing the partial sums stored in the Sierpinski tree for $N = 9$.}
    \label{fig:sierpinski_nx}
\end{figure}

Now, consider the data structure defined by the $N$-node tree $T(N)$. For a bit corresponding to the node $j \in V(N)$, let the set of bits flipped by the array update operation be $B_N(j)$, and the set of bits summed for the prefix sum operation be $P_N(j)$. We want to show that these sets are $O(\log_3 N)$ in size. We will do this by showing that the size of their union is $O(\log_3 N)$. To this end we define the \emph{weight} of each node in the tree to be 
\begin{equation}\label{eq:weight}
w_N(j) \equiv \abs{B_{N}(j) \cup P_{N}(j)}.
\end{equation}
We will need the following lemma.

\begin{lemma}\label{pf:lemma}
$w_N(j)$ is an increasing function of $N$.
\end{lemma}
\begin{proof}
It is sufficient to show that removing nodes from a Sierpinski tree in descending order does not increase the cost of the array update and prefix sum operations associated with any of the remaining nodes. Without loss of generality, we may assume we start from some full tree $T(3^k)$, where $k \in \mathbb{N}$.

Let the left, central and right subtrees of $T(3^k)$ be $T_L(3^k)$, $T_C(3^k)$ and $T_R(3^k)$, with roots $L$, $C$ and $R$ respectively. We will now consider what happens when we remove the nodes of the tree in descending order.

It is easily seen from the structure of the tree that deleting nodes in $T_R$ does not affect the array update operation for nodes in $T_L$ and $T_C$, which have no ancestors in $T_R$. Similarly, the only effect on the prefix sum operation for nodes in $T_L$ and $T_C$ is to possibly reduce the number of bits summed if the connection to the root $C$ of $T_C$ is removed.

Let us suppose we have deleted $T_R$ entirely. Now deleting nodes in $T_C$ does not increase the array update or prefix sum cost for nodes in $T_L$ (only reducing the former if $C$ is deleted). If it is the case that removing nodes from a subtree likewise does not increase the costs for other nodes in that subtree, then we are done.

We are now able to prove the lemma by induction. For a tree of $N \leq 3$ nodes, by inspection, removing nodes from the tree in descending order does not increase costs for other nodes in the tree. Now suppose that there exists some integer $M \geq 1$ such that removing nodes in descending order from $T(3^{M})$ does not increase operation costs for other nodes in the tree. Then for a tree $T(3^{M+1})$, we have the following:
\begin{enumerate}
\item Removing nodes in descending order from $T_R(3^{M+1})$ does not increase costs for nodes in $T_L(3^{M+1})$ and $T_C(3^{M+1})$.
\item Removing nodes in descending order from $T_C(3^{M+1})$ does not increase costs for nodes in $T_L(3^{M+1})$.
\item By the inductive hypothesis, removing nodes in descending order from any subtree $T_X$ does not increase costs for nodes in that subtree. More precisely, the size of $\left(B_N(j) \cup P_N(j)\right) \cap V_X$ does not increase.
\end{enumerate}
Altogether, we have that removing nodes from $T(3^{n+1})$ in descending order does not increase costs for any of the remaining nodes. This completes the proof by induction.
\end{proof}
We are now ready to prove our main result.

\begin{theorem}
For an $N$-node Sierpinski tree $T(N)$, $\forall j \in V(N)$, $\omega_N(j) \leq \lceil \log_3 N \rceil + 1$, with equality when $N$ is a power of 3.
\end{theorem}

\begin{proof}
We begin by showing that for a full tree, $w_N(j) = \log_3 N + 1$ for all $j$. We will do this by induction.

By inspection, $w_3(j) = \log_3 3 + 1 = 2$ for all $j$.

Now suppose that there exists some integer $k \geq 1$ for which 
\begin{equation}\label{eq:inductiveHypothesis}
w_{3^k}(j) = k + 1
\end{equation}
for all nodes $j$. We will show that this implies 
\begin{equation}\label{eq:inductiveStep}
w_{3^{k+1}}(j) = k + 2
\end{equation}
for all nodes $j$. To this end, consider the tree with $3^{k+1}$ nodes, $T(3^{k+1}$). We divide this into left, center, and right subtrees,
\begin{equation}
\begin{aligned}
T_L &= T\left(0,\ 3^{k}\right) \equiv \{V_L, E_L\},\\
T_C &= T\left(3^k, 3^k\right) \equiv \{V_C, E_C\},\\
T_R &= T\left(2\cdot 3^k,\ 3^{k}\right) \equiv \{V_R, E_R\},
\end{aligned}
\end{equation}
with roots $L$, $C$ and $R$, respectively.

By the inductive hypothesis \eqref{eq:inductiveHypothesis}, for each subtree $T_X$, we have $w_{3^k}(j) = k + 1$ for all $j$ if we treat the subtree ``independently''. More formally,
\begin{equation}
\forall j \in V_X(3^{k+1}),\ \abs{\left(B_{3^{k+1}(j)}\cup P_{3^{k+1}}(j)\right)\cap V_X(3^{k+1})} = k + 1.
\end{equation}
If we connect them into a single tree, the array update cost for each node in $T_L$ increases by 1, as each update now also flips bit $C$. The prefix sum cost for $T_L$ is unchanged, as none of the nodes in $T_C$ or $T_R$ have indices less than those of the nodes in $T_L$. Altogether, for nodes $j \in V_L$, $w_{3^{k+1}}(j) = k + 2$.

The array update and prefix sum costs for nodes in tree $T_R$ each increase by 1, as node $C$ must now be flipped when any bit in $T_R$ is flipped, and must likewise be counted in all prefix sums. As the value of bit $R$ has been added to bit $C$, bit $R$ must also be counted in prefix sums. However, bit $R$ was already involved in all value updates. Thus, the bit $C$ is the only new bit that matters, and $w_{3^{k+1}}(j) = k + 2$ for all $j$ in $V_R$.

For $T_C$, the value update operation is unchanged. For the prefix sum operation, we must treat nodes with indices $j \leq C$ and indices $j > C$ separately. For nodes with indices $j \leq C$, the prefix sum operation must now count bit $L$. For nodes with indices $j > C$, the prefix sum operation must count $R$, as the value of node $R$ has been added to node $C$. Altogether, $w_{3^{k+1}}(j) = k + 2$ for all $j$ in $V_C$.

Overall, we have shown that \eqref{eq:inductiveHypothesis} implies \eqref{eq:inductiveStep}, and hence \eqref{eq:inductiveHypothesis} holds for all $k \in \mathbb{N}$ by induction. 

When $N$ is not a power of 3, i.e.
\begin{equation}\label{eq:logk}
\begin{aligned}
&3^k < N < 3^{k+1}\\
\iff &k < \log_3 N < k + 1,
\end{aligned}
\end{equation}
we can show that since $w_{N}(j)$ is an increasing function of $N$ (Lemma \ref{pf:lemma}), $\forall j \in V(N) \subset V(3^{k+1})$, 
\begin{equation}
\begin{aligned}
&w_{N}(j) \leq w_{3^{k+1}}(j)\\
\implies &w_{N}(j) \leq k + 2\\
\implies &w_{N}(j) \leq \lceil \log_3 N\rceil + 1,
\end{aligned}
\end{equation}
where in the last line we have used \eqref{eq:logk}. This completes the proof.
\end{proof}

\section{Discussion}

We have presented the Sierpinski tree, a novel data structure similar to the Fenwick tree, but with strictly better performance. We will now very briefly discuss the application of this tree to quantum computing, and comment on its near-optimality.

The Fenwick tree can be used in the context of quantum computation to define a fermion-to-qubit transform~\cite{havlicek_operator_2017}. Without going into the details, this amounts to a prescription for associating with each node of the tree two tensor products of Pauli matrices (``Pauli strings''). The ``Pauli weight'' of each of these strings is defined as the number of non-identity Pauli matrices in the string. For the purpose of quantum simulation, it is desirable to minimize this quantity~\cite{havlicek_operator_2017,bravyi_fermionic_2002}.

In \cite{havlicek_operator_2017}, Havlicek et al.\ show that the worst-case Pauli weight of the Pauli strings associated with the $j$th node of an $N$-node Fenwick tree is, using our terminology (Eq.~\ref{eq:weight}), that node's weight $w_N(j)$. For a Fenwick tree, this weight is $O(\log_2 N)$. It is similarly possible to define a fermion-to-qubit transform associated with the Sierpinski tree, for which the worst-case Pauli weight of each Pauli string associated with the $j$th node would be $w_N(j) \leq \lceil \log_3 N \rceil + 1$.

Interestingly, in \cite{jiang_optimal_2020}, Jiang et al.\ show that the minimum average Pauli weight for a fermion-to-qubit transform of this type is $\log_3 (2N) = \log_3 N + \log_3 2$. This implies in turn that the best possible Fenwick tree-like structure would have average weight $\overline{w}_N$ bounded by
\begin{equation}
\log_3 N + \log_3 2 \leq \overline{w}_N \leq \lceil \log_3 N \rceil  + 1 .
\end{equation}
The Sierpinski tree is therefore very nearly optimal by this metric, though there is room to improve on the version presented here. For example, removing the edge between nodes 13 and 22 in Fig.~\ref{fig:n27} reduces the weight of nodes 14 through 17 by one without increasing any other weights. More generally, one can iterate over the edges of the tree, deleting them if doing so would reduce the average weight of the nodes, and obtain a slight improvement.

In a companion paper, the authors will fully describe the fermion-to-qubit transform defined by the Sierpinski tree.

\section*
{Acknowledgements}
The authors thank Peter Winkler, Ojas Parekh, and Joseph Gibson for helpful discussions. The authors were supported by the US NSF grant PHYS-1820747. Additional support came from the NSF (EPSCoR-1921199) and from the Office of Science, Office of Advanced Scientific Computing Research under program Fundamental Algorithmic Research for Quantum Computing. This paper was also supported by the ``Quantum Chemistry for Quantum Computers'' project sponsored by the DOE, Award DE- SC0019374. JDW holds concurrent appointments at Dartmouth College and as an Amazon Visiting Academic. This paper describes work performed at Dartmouth College and is not associated with Amazon.

\bibliographystyle{unsrt}
\bibliography{citations}

\end{document}